\documentclass{amsart}
\usepackage{times}
\usepackage{amssymb,amsmath,amscd, amsthm,stmaryrd,verbatim, mathrsfs}

\usepackage{color}

\newtheorem{Thm}{Theorem}
\newtheorem{theorem}{Theorem}[section]

\newtheorem{Lem}[theorem]{Lemma}

\numberwithin{equation}{section}

\newcommand{\mb}{\mathbf}
\newcommand{\bb}{\mathbb}
\newcommand{\ms}{\mathscr}
\newcommand{\mr}{\mathrm}
\newcommand{\frk}{\mathfrak}

\usepackage{hyperref}
\begin{document}

\title{On the orbits of the magnetized Kepler problems in dimension $2k+1$ }

\author{Zhanqiang Bai, Guowu Meng and Erxiao Wang}

\address{Department of Mathematics, Hong Kong Univ. of Sci. and
Tech., Clear Water Bay, Kowloon, Hong Kong}

\email{mazqbai@ust.hk}

\address{Department of Mathematics, Hong Kong Univ. of Sci. and
Tech., Clear Water Bay, Kowloon, Hong Kong}

\email{mameng@ust.hk}

\address{ Wuhan Institute of Physics and Mathematics, Chinese Academy of Sciences}

\email{wangex@wipm.ac.cn}

\thanks{The  authors were all supported by the Hong Kong Research Grants Council under RGC Project No. 603110 and the Hong Kong University of Science and Technology under DAG S09/10.SC02.}
\thanks{The third author was also supported by the NSF of China (Grant Nos. 10941002, 11001262), and the Starting Fund for Distinguished Young Scholars of Wuhan Institute of Physics and Mathematics (Grant No. O9S6031001).
 }


\date{February 28, 2013}


\maketitle
\begin{abstract}
It is demonstrated that, for the recently introduced classical magnetized Kepler problems in dimension $2k+1$, the non-colliding orbits in the ``external configuration space" $\mathbb R^{2k+1}\setminus\{\mathbf 0\}$ are all conics, moreover,  a conic orbit is an ellipse, a parabola, and a branch of a hyperbola according as the total energy is negative, zero, and positive. It is also demonstrated that the Lie group ${\mr {SO}}^+(1,2k+1)\times {\bb R}_+$ acts transitively on both the set of oriented elliptic orbits and the set of oriented parabolic orbits. \end{abstract}

\tableofcontents

\section {Introduction}
The Kepler problem for planetary motions is a two-body dynamic problem with an attractive force obeying the inverse square law. Mathematically it can be reduced to the one-body dynamic problem with the equation of motion
\begin{eqnarray}\label{eqnofM}
{\mb r}''= -{{\mb r}\over r^3},
\end{eqnarray}
where ${\mb r}$ is a function of $t$ taking value in ${\bb R}^3_*:={\bb R}^3\setminus\{{\mb 0}\}$, ${\mb r}''$ is the acceleration vector and $r$ is the length of $\mb r$.  We shall refer to this later one-body dynamic problem as the {\em Kepler problem}.

\vskip 10pt
A surprising discovery due to D. Zwanziger \cite{Z68} and to H. McIntosh and A. Cisneros \cite{MC70} independently in the late 1960s is that there exist magnetized companions for the Kepler problem. These extra dynamic problems plus the Kepler problem, referred to as {\em MICZ-Kepler problems}, are indexed by the magnetic charge $\mu$, with $\mu=0$ for the Kepler problem. The parameter $\mu$ can take any real number at the classical mechanics level, a half of any integer at the quantum mechanics level. We would like to point out that, a MICZ-Kepler problem is the large distance limit of a system in which a spinless test particle moving outside the core of a self-dual $\mathrm{SU}(2)$ monopole, see Ref. \cite{Cordani90} and references therein.  

Since the Kepler problem has long been known to exist in all dimensions, one naturally wonders whether there are magnetized Kepler problems in higher dimensions. By realizing \cite{Iwai85} that the MICZ-Kepler problems are the $\mathrm{U}(1)$-symmetric reductions of the isotropic oscillator in space $\mathbb R^4=\mathbb C^2$, T. Iwai's \cite{Iwai90} obtained the magnetized Kepler problems in dimension five (referred to as the $\mathrm{SU}(2)$-Kepler problems), as the $\mathrm{SU}(2)$-symmetric reductions of the isotropic oscillator in space $\mathbb R^8=\mathbb H^2$. For quite a while the magnetized Kepler problems were thought to exist only in dimensions three, five and (possibly) nine, corresponding to the division algebras $\mathbb C$, $\mathbb H$ and $\mathbb O$ respectively. 

Recently, a family of Poisson realizations of $\mathfrak{so}(2, 2k+1)$ was found by the second-named author \cite{meng2012b} for any integer $k\ge 1$; as a by-product, there is now a family of the magnetized Kepler problems at the classical mechanics level in all odd dimensions. These magnetized Kepler problems are the MICZ-Kepler problems in dimension three and Iwai's $\mathrm{SU}(2)$-Kepler problems in dimension five. We would like to point out that, in dimension five or above, a complication arises: for $k>1$, when the magnetic charge is non-zero, instead of being $X:=\mathbb R^{2k+1}\setminus\{\mathbf 0\}$, the configuration space $P$ fibers over $X$, with the fiber being diffeomorphic to $\mathrm{SO}(2k)/\mathrm{U}(k)$. 

\vskip 10pt
While the orbits for the magnetized Kepler problems in dimension three have been thoroughly studied from the very beginning \cite{MC70}, in higher dimensions, in view of the fact that the equation of motion is a bit more sophisticated, one might expect that the orbits are a bit more sophisticated, too. That is probably the reason why the orbits for Iwai's $\mathrm{SU}(2)$-Kepler problems were never investigated in Ref. \cite{Iwai90} and the subsequent papers \cite{Iwai96}.

The main purpose of this article is to investigate the orbits for the magnetized Kepler problems in all odd dimensions. Since the bundle $P\to X$ has a canonical connection, to understand an orbit $C$ inside $P$, we just need to understand its projection $\bar C$ onto $X$ and then obtain the orbit $C$ as the lifting (via the canonical connection) of $\bar C$ to $P$. The hard part of this strategy is the realization that $\bar C$ must be a conic if $C$ is a non-colliding orbit. The breakthrough came when we found the effective angular momentum $\bar L$, which is a constant $2$-vector such that $\bar L\wedge \bar L =0$ and $\bar L$ reduces to the angular momentum $L$ when either the dimension is three or the magnetic charge is zero. With the orbits well-understood, we then extend the main result of Ref. \cite{meng2012} beyond dimension three. 

The important issue of geometric quantization, in the spirit of a series of papers \cite{Mladenov} by I. Mladenov and V. Tsanov, shall be carried out in the future.

\subsection{Outline} In section \ref{S: Kepler}, we quickly review the Kepler problem in arbitrary dimensions, with a focus on the non-colliding orbits and their light cone reformulation. In the next three sections, we review the magnetized Kepler problems in arbitrary odd dimensions, firstly the equation of motion in section \ref{S: MICZ}, then a very useful lemma in section \ref{S: Monopole}, and finally the total energy, the angular momentum and the Lenz vector in section \ref{S: Lenz}. In the remaining two sections, we state and prove our main results, cf. Theorem \ref{maintheorem1} in section \ref{S: Main1} and Theorem \ref{maintheorem2} in section \ref{S: Main2}.  

 \subsection{Notations and Conventions}
We are mainly dealing with poly-vectors in the euclidean space $\mathbb R^n$ or the Lorentz space $\mathbb R^{1, n}$, plus the wedge product, interior product, and inner product involving the poly-vectors. 

The boldface Latin letters are reserved for vectors in the euclidean space $\mathbb R^n$ only, and the inner product (i.e., the dot product) of the vectors $\mathbf A$ and $\mathbf B$ is written as $\mathbf A\cdot \mathbf B$. Vectors in the Lorentz space $\mathbb R^{1, n}$ are referred to as {\em Lorentz vectors}. For the Lorentz vectors $a=(a_0, \mathbf a)$ and $b=(b_0, \mathbf b)$,
the (Lorentz) inner product of $a$ and $b$, written as $a\cdot b$, is defined as $$a_0b_0-\mathbf a\cdot \mathbf b.$$

The vector $\mathbf r$ is reserved for a point in $\mathbb R^n$, and the Lorentz vector $x$ is reserved for a point in $\mathbb R^{1, n}$. We often write $x=(x_0, \mathbf r)$ rather than $(x_0, \mathbf x)$. For the standard basis vectors  
$e_0$, $e_1$, \ldots, $e_n$ in $\mathbb R^{1, n}$, we have $e_0\cdot e_0=1$ and $e_i\cdot e_i=-1$ for $i>0$. When we view the Lorentz $e_i$ ($i>0$) as a vector inside the subspace $\mathbb R^n$, we write it as $\mathbf e_i$.

Let $V$ be either $\mathbb R^n$ or $\mathbb R^{1, n}$, $k>0$ be an integer. A $k$-vector in $V$ is just an element of $\wedge^k V$. An 1-vector is just a vector. A $k$-vector is called {\em decomposable} if it is the wedge product of $k$ vectors. It is a trivial fact that a $2$-vector in a 3D space is always decomposable. In case $X$ is a non-zero decomposable $k$-vector in $V$, we use $[X]$ to denote the $k$-dimensional oriented subspace of $V$, with $X$ representing its orientation.

The inner product extend from vectors to poly-vectors and is denoted by $\langle\,, \rangle$. By definition, for vectors $u_1$, \ldots, $u_k$, $v_1$, \ldots, $v_k$ in $V$, let $[u_i\cdot v_j]$ be the square matrix whose $(i,j)$-entry is $u_i\cdot u_j$, then
$$
\langle u_1\wedge\cdots\wedge u_k, v_1\wedge\cdots\wedge v_k\rangle =\det [u_i\cdot v_j]. 
$$ 

We define the interior product $\lrcorner$ as the adjoint of the wedge product with respect to the inner product for poly-vectors: for poly-vectors $X$, $u$ and $v$ in $V$ with $\deg X+\deg u=\deg v$, we have
$$
\langle X\wedge u, v\rangle= \langle u, X\lrcorner v\rangle.
$$

For poly-vector $X$, we write $X^2$ for $\langle X, X \rangle$. In case $\langle X, X \rangle\ge 0$, we write $|X|$ for $\sqrt{\langle X, X\rangle}$. We always write $r$ for $|\mathbf r|$. Finally, we remark that a poly-vector $X$ in $\mathbb R^n$ is also viewed as a poly-vector $X$ in $\mathbb R^{1,n}$ in a natural way. 

\section{The Kepler problem in dimension $n$}\label{S: Kepler}
Taking the equation of the motion (\ref{eqnofM}) and replacing ${\bb R}^3_*$ by ${\bb R}^n_*:={\bb R}^n\setminus\{{\mb 0}\}$ with $n\ge 2$, we get an analogue of the Kepler problem in dimension $n$ (referred to as the Kepler problem in dimension $n$) in the sense that these analogues share the same characteristic feature as the Kepler problem:  existence of the Lenz vector.  Indeed, if we introduce
$$L:=\mb r\wedge \mb r', \quad \mb A := \mb r' \lrcorner L+{\mb r\over r},$$
one can easily check that both the angular momentum $L$ and the Lenz vector $\mb A$ are constants of motions. 

An orbit is a colliding orbit if and only if $L=0$. Note that we are only interested in the non-colliding orbits, so we assume hereon that $L\neq 0$. Since $L$ is decomposable, $L$ determines an oriented 2D subspace $[L]$ of $\mathbb R^n$. It is easy to see that $\mathbf A\wedge L=0$ and
\begin{eqnarray}\label{conics}
\mathbf r\wedge L=0, \quad r-\mathbf A\cdot \mathbf r=|L|^2.
\end{eqnarray}
Therefore, the non-colliding orbits are oriented conics inside the oriented 2D subspace $[L]$, with $0$ as a focus and  $\mathbf A$ as its eccentricity vector; moreover, one can show that the total energy for a motion with such an orbit is
$$
E=-{1-|\mathbf A|^2\over 2|L|^2}.
$$ Note that the orbit is oriented in the sense that $\mathbf t\wedge \mathbf n$ is a positive multiple of $L$. Here, $\mathbf t$ and $\mathbf n$ are respectively the unit tangent vector and unit normal vector of the oriented curve.

\subsection{The reformulation in the Lorentz space $\mathbb R^{1, n}$}
The Kepler problem in dimension $n$ has a mathematically appealing reformulation in the Lorentz space $\mathbb R^{1, n}$, as originally pointed out in Ref. \cite{meng2013}. To see this, we observe that $\mathbb R^n_*$ is diffeomorphic to the future light cone
$$
\{x\in \mathbb R^{1, n}\mid x^2=0, x_0>0\},
$$
so the Kepler problem in dimension $n$ can be reformulated as a dynamic problem on the future light cone. As a result, an oriented orbit in $\mathbb R^n_*$ can be reformulated as an oriented curve inside this future light cone, i.e., the intersection of the cylinder over the oriented orbit with the future light cone. It turns out this reformulated orbit is a conic section, and this gives another explanation why a non-colliding orbit must be a conic.

To see clearly the intersection plane for the conic section, we write $A$ for $(1, \mathbf A)$ and observe that $L\wedge A=L\wedge e_0\neq 0$.
For $x$ on the future light cone, Eq. (\ref{conics}) can be recast as
\begin{eqnarray}\label{conics1}
x\wedge (L\wedge A)=0, \quad A\cdot x =|L\wedge A|^2.
\end{eqnarray}  
Eq. (\ref{conics1}) defines an affine plane and its intersection with the future light cone is the same as the intersection of the cylinder over the orbit with the future light cone.

\section{Magnetized Kepler problems in Dimension $n=2k+1$}\label{S: MICZ}
The purpose of this section is to describe the equation of motion for the magnetized Kepler problems in odd dimension $n\ge 3$, as first appeared in Ref. \cite{meng2012b}. This equation is the $n$-dimensional analogue of the equation of motion
\begin{eqnarray}\label{EM}
{\mb r}'' = -{\mb r\over r^3}+\mu^2{\mb r\over r^4}- {\mb r}'\times \mu{\mb r\over r^3}
\end{eqnarray} for the MICZ-Kepler problems \cite{MC70, Z68}, where the parameter $\mu$ is the magnetic charge.

It turns out that the high dimensional analogue of Eq. (\ref{EM}) is far from straightforward. The reason is that, if $k>1$, instead of governing motions on $\mathbb R^{2k+1}_*$, the equation of motion governs motions on a manifold $P_\mu$ which fibers over $\mathbb R^{2k+1}_*$.

\vskip 5pt
To describe the fiber bundle $P_\mu\to \mathbb R^{2k+1}_*$, we let $G=\mathrm{SO}(2k)$ and consider the canonical principal $G$-bundle over $\mathrm{S}^{2k}$:
$$\begin{array}{c}
 \mathrm{SO}(2k+1)\cr
 \Big\downarrow\cr
 \mathrm{S}^{2k}.
\end{array}
 $$This bundle comes with a natural connection
$$\omega (g):=\mathrm{Pr}_{\frk{so}(2k)}\left(g^{-1}dg\right),$$
where $g^{-1}dg$ is the Maurer-Cartan form for $ \mathrm{SO}(2k+1)$, so it is an $\frk{so}(2k+1)$-valued differential one form on $\mathrm{SO}(2k+1)$, and $\mathrm{Pr}_{\frk{so}(2k)}$ denotes the orthogonal projection of $\frk{so}(2k+1)$ onto $\frk{g}:=\frk{so}(2k)$.

Under the map
\begin{eqnarray}
\pi:  \mathbb R^{2k+1}_* & \to & \mathrm{S}^{2k}\cr
 \mb  r &\mapsto &{\mb r \over r},
\end{eqnarray}
the above bundle and connection are pulled back to a principal $G$-bundle
\begin{eqnarray}
\begin{array}{l}
P\cr
 \Big\downarrow\cr
X:=\mathbb R^{2k+1}_*
\end{array}
\end{eqnarray} with a connection which is usually referred to as the generalized Dirac monopole \cite{monopole}. Now $$P_\mu\to \mathbb R^{2k+1}_*$$ is the associated fiber bundle with fiber being a certain co-adjoint orbit $\mathcal O_\mu$ of $G$, the so-called magnetic orbit  with magnetic charge $\mu\in \mathbb R$.

To describe $\mathcal O_\mu$, let us use $\gamma_{ab}$ ($1\le a, b\le 2k$) to denote the element of $i\frk{g}$ such that in the defining representation of $\frk{g}$, $M_{a,b}:=i\gamma_{ab}$ is represented by the skew-symmetric real symmetric matrix whose $a b$-entry is $-1$, $b a$ entry is $1$, and all other entries are $0$. For the invariant metric $(, )$ on $\frk{g}$, we take the one such that $M_{a,b}$ ($1\le a< b\le 2k$) form an orthonormal basis for $\frk g$. Via this invariant metric, one can identify $\frk{g}^*$ with $\frk{g}$, hence co-adjoint orbits with adjoint orbits. By definition, for any $\mu\in \mathbb R$, 
\begin{eqnarray}
\mathcal O_\mu := \mr{SO}(2k)\cdot {1\over \sqrt k}(|\mu|M_{1,2}+\cdots + |\mu| M_{2k-3, 2k-2} + \mu M_{2k-1, 2k} ).
\end{eqnarray} It is easy to see that $\mathcal O_\mu=\{0\}$ if $\mu=0$ and is diffeomorphic to $\mr{SO}(2k)\over \mr{U}(k)$ if $\mu\neq 0$.

\vskip 5pt
We are now ready to describe the equation of motion for the magnetized Kepler problem in dimension $2k+1$. Let $\mathbf r$: $\mathbb R\to X$ be a smooth map, and $\xi$ be a smooth lifting of $\mathbf r$:
\begin{eqnarray}
\begin{array}{ccc}
 & & P_\mu\\
 \\
 & \xi \nearrow &\Big\downarrow\\
 \\
 \mathbb R &\buildrel\mathbf r\over\longrightarrow & X
\end{array}
\end{eqnarray}
Let $Ad_P$ be the adjoint bundle $P\times_G\frk{g}\to X$, $d_\nabla$ be the canonical connection, i.e., the generalized Dirac monopole on $\mathbb R^{2k+1}_*=X$. Then the curvature $\Omega:=d_\nabla^2$ is a smooth section of the vector bundle $\wedge^2T^*X\otimes Ad_P$. (With a trivialization of $P\to X$, locally $\Omega$ can be represented by ${1\over 2}\sqrt{-1}F_{jk}\, dx^j\wedge dx^k$.)
The equation of motion is
\begin{eqnarray}\label{EqnMF}
\fbox{$\left\{
\begin{array}{l}
\mathbf r'' = -{\mathbf r\over r^3}+{\mu^2\over k} {\mathbf r\over r^4}+ (\xi,\mathbf r'\lrcorner \Omega),\cr
\\
{D\over dt}\xi =0.
\end{array}
\right.$}
\end{eqnarray} Here ${D\over dt}\xi$ is the covariant derivative of $\xi$, $(, )$ refers to the inner product on the fiber of the adjoint bundle coming from the invariant inner product on $\frak g$, and 2-forms are identified with 2-vectors via the standard euclidean structure of $\mathbb R^{2k+1}$. Eq. (\ref{EqnMF}) defines a super integrable model, referred to as the {\bf classical Kepler problem with magnetic charge $\mu$ in dimension $2k+1$}, which generalize the classical MICZ-Kepler problem. Indeed, in dimension $3$, the bundle is topological trivial, $\xi=\mu M_{12}$, and $\Omega = {*(\sum_{i=1}^3 x^i\, dx^i)\over r^3} M_{12}$, then Eq. (\ref{EqnMF}) reduces to Eq. (\ref{EM}), i.e., the equation of motion for the MICZ-Kepler problem with magnetic charge $\mu$.  In dimension $5$, it is essentially Iwai's $\mr{SU}(2)$-Kepler problem, cf. Ref. \cite{Iwai90}.

The equation of motion appears to be mysterious, but it doesn't. As demonstrated in Ref. \cite{meng2012b}, with a key input from the work of Sternberg \cite{Sternberg77}, Weinstein \cite{Weinstein78}, and Montgomery \cite{Montgomery84}, it emerges naturally from the notion of universal Kepler problem in Ref. \cite{meng11a}. As a side remark, we would like to point out that the quantum magnetized Kepler problems \cite{meng07} were obtained much earlier. 

\section{A useful lemma on the generalized Dirac monopoles}\label{S: Monopole}
This article as well as Refs. \cite{meng07, meng2012b} crucially depend on a lemma about the generalized Dirac monopoles \cite{monopole}. The purpose of this section is to introduce an abridged version of this lemma as used in Ref. \cite{meng2012b}.

We shall write $\mathbf r =(x^1, \ldots, x^n)$ for a point in $\mathbb R^n_*$ and $r$ for the length of $\mathbf  r$. The small Lartin letters $j$, $k$, etc. be indices that run from $1$ to $n$, and the small Latin letters $a$, $b$, etc. be indices that run from $1$ to $n-1$. To do local computations, we need to choose a bundle trivialization on $U$ which is $\mathbb R^n$ with the negative $n$-th axis removed and then write down the gauge potential explicitly. Indeed, there is a bundle trivialization on $U$ such that the gauge potential $A=A_k\, dx^k$ can be written as
\begin{eqnarray}\label{mnple}
A_n =0,\quad A_b=-{1\over r(r+x_n)}x^a\gamma_{ab}
\end{eqnarray}
where $x^a\gamma_{ab}$ means $\sum_{a=1}^{n-1}x^a\gamma_{ab}$, something we shall assume whenever there is a repeated index.  It is then clear that the gauge field strength
$F_{jk}: = \partial_jA_k-\partial_k A_j + i [A_j, A_k]$ is of the form
\begin{eqnarray}\label{Ffield}
F_{nb} = {1\over r^3}x^a\gamma_{ab},\quad
F_{ab} = -{1\over r^2}(\gamma_{ab} +x^aA_b - x^bA_a).
\end{eqnarray}

The following lemma is an abridged version of Lemma 4.1 in Ref. \cite{meng2012b}.
\begin{Lem}\label{lemma}
Assume $n=2k+1$. Let $Q={1\over 2k}\sum_{a, b}(\gamma_{ab})^2$ and $\nabla_k=\partial_k+iA_k$. For the gauge potential $A$ defined in Eq. (\ref{mnple}),  the following statements are true.

1) As $i\,\frk{so}(2k)$-valued functions on $U$,
\begin{eqnarray}\label{Id1}
x^k A_k=0,\quad x^j F_{jk}=0,\quad \nabla_l F_{jk}={1\over r^2}\left( -x^j
F_{lk}-x^k F_{jl}-2x^l F_{jk} \right).
\end{eqnarray}
Consequently, we have $\mathbf r\lrcorner F=0$ and ${D\over dt} (r^2F)=-\mathbf r\wedge (\mathbf r'\lrcorner F)$.

2) Assume $\xi\in \mathcal O_\mu\subset \frk{g}$. As real functions on $U$,
\begin{eqnarray}\label{Id3}
r^4\sum_k( i\xi, F_{kj})(i \xi, F_{kj'})={\mu^2\over k}\left(\delta_{jj'}-{x^j x^{j'}\over
r^2}\right).
\end{eqnarray}  Consequently, we have $|(i\xi, r^2F)|^2=\mu^2$ and $$(i\xi, \mathbf r' \lrcorner F)\lrcorner (i\xi, r^2F)=-{\mu^2/k\over r}\left({\mathbf r\over r}\right)', \quad |r^2(i\xi, \mathbf r' \lrcorner F)|^2={\mu^2\over k}{|\mathbf r\wedge \mathbf r'|^2\over r^2}.$$
\end{Lem}

\section{The angular momentum and the Lenz vector}\label{S: Lenz}
Besides the equation of motion (\ref{EqnMF}), Ref. \cite{meng2012b} also provides formulae for the total energy
\begin{eqnarray}\label{hamiltonian}
E={1\over 2}|\mathbf r'|^2-{1\over r}+{\mu^2/k\over 2r^2},
\end{eqnarray} the angular momentum
\begin{eqnarray}
L=\mathbf r\wedge \mathbf r'+ (\xi,r^2\Omega)
\end{eqnarray} and the Lenz vector
\begin{eqnarray}
\mathbf A= \mathbf r' \lrcorner L+{\mathbf r\over r}.
\end{eqnarray}

To verify that $L$ and $\mathbf A$ are constants of motion directly, it suffices to do local computations over $U$. That is because $U$ is $\mathbb R^{2k+1}$ with a so called Dirac string (i.e., the negative $n$-th coordinate axis) removed, and, for a given non-colliding orbit, by dimension reason we may assume it misses the entire $n$-the coordinate axis.

Over the dense open set $U$, the bundle $P\to X$ shall be trivialized in a way so that the gauge potential is of the form as in equation (\ref{mnple}), so that the curvature $\Omega$ shall be represented by $iF$ and the lifting $\xi$ shall be represented by a smooth map (also denoted by $\xi$) from $\mathbb R$ into $\frk g$ whose image is always inside $\mathcal O_\mu$.  With this understood, one can verify directly that $L$ and $\mathbf A$ are constants of motions. For example, use the lemma and the equation of motion, since $L=\mathbf r\wedge \mathbf r'+(i\xi, r^2F)$, we have
\begin{eqnarray}
L' &= & \mathbf r\wedge \mathbf r''+(i\xi, -\mathbf r\wedge (\mathbf r' \lrcorner F))\cr
&= & \mathbf r\wedge (\mathbf r''-(\xi, \mathbf r' \lrcorner \Omega))\cr
&=&  \mathbf r\wedge \left(-{\mathbf r\over r^3}+{\mu^2\over k} {\mathbf r\over r^4}\right) =0,\nonumber
\end{eqnarray}
similarly, we have
\begin{eqnarray}
\mathbf A' &= & \mathbf r''\lrcorner L+\left({\mathbf r\over r}\right)'\cr
&=&  \left(-{\mathbf r\over r^3}+{\mu^2\over k} {\mathbf r\over r^4}+ (i\xi,\mathbf r'\lrcorner F)\right)\lrcorner \left(\mathbf r\wedge \mathbf r'+ (i\xi,r^2F)\right)+\left({\mathbf r\over r}\right)'\cr
&=&\left (-{1\over r^3}+{\mu^2\over k} {1\over r^4}\right) \mathbf r\lrcorner (\mathbf r\wedge \mathbf r')+ (i\xi,\mathbf r'\lrcorner F)\lrcorner (i\xi,r^2F)+\left({\mathbf r\over r}\right)'\cr
&=&\left (-1+{\mu^2\over k} {1\over r}\right) \left({\mathbf r\over r}\right)'-{\mu^2/k\over r}\left({\mathbf r\over r}\right)'+\left({\mathbf r\over r}\right)'\cr
&=& 0.\nonumber
\end{eqnarray}

For a non-colliding orbit, since $\mathbf r\wedge \mathbf r'\neq 0$, use the lemma, we have
\begin{eqnarray}\label{Lmu}
|L|^2= |\mathbf r\wedge \mathbf r'|^2+ |(i\xi, r^2F)|^2=  |\mathbf r\wedge \mathbf r'|^2+ \mu^2> \mu^2.
\end{eqnarray}
It has also been shown in Ref. \cite{meng2012b} that, for a non-colliding orbit,  the total energy is completely determined by the angular momentum $L$ and the Lenz vector $\mathbf A$:
\begin{eqnarray}\label{formulaE}
E=-{1-|\mathbf A|^2\over 2(|L|^2-\mu^2)}.
\end{eqnarray}
To see this, we note that
\begin{eqnarray}
|\mathbf A|^2 &=& 1+{2\over r}\mathbf r \lrcorner(\mathbf r' \lrcorner L)+|\mathbf r' \lrcorner L|^2\cr
&=& 1+{2\over r} \langle\mathbf r'\wedge \mathbf r, L\rangle+|(\mathbf r\cdot \mathbf r')\mathbf r'-(\mathbf r'\cdot \mathbf r')\mathbf r +r^2(i\xi, \mathbf r' \lrcorner F)|^2\cr
&=& 1-{2\over r} |\mathbf r\wedge \mathbf r'|^2+|(\mathbf r\cdot \mathbf r')\mathbf r'-(\mathbf r'\cdot \mathbf r')\mathbf r |^2+|r^2(i\xi, \mathbf r' \lrcorner F)|^2\cr
&=& 1-{2\over r} |\mathbf r\wedge \mathbf r'|^2+|\mathbf r'|^2 |\mathbf r\wedge \mathbf r'|^2+{\mu^2\over k}{|\mathbf r\wedge \mathbf r'|^2\over r^2}\cr
&=& 1+2|\mathbf r\wedge \mathbf r'|^2 E \quad \mbox{By Eq. (\ref{hamiltonian})}\cr
&= & 1+2(|L|^2-\mu^2)E.\quad \mbox{By identity (\ref{Lmu})}\nonumber
\end{eqnarray}

\section{The MICZ-Kepler orbits are all conics}\label{S: Main1}
In this and next sections we shall study the orbits of magnetized Kepler problems in a {\em fixed} dimension $2k+1$, but with arbitrary magnetic charges. In dimension higher than $3$, the angular momentum is no longer decomposable, and that makes our work in higher dimensions more sophisticated than it might appear. 

A solution $(\mathbf r (t), \xi(t))$ to the equation of motion (\ref{EqnMF}) shall be referred to as a {\em motion}, whose total trace inside $P_\mu$ shall be referred to as an {\em orbit}. Under the bundle projection $\pi_{\mu}: P_\mu\to\mathbb R^{2k+1}_*$, these orbits become curves inside the ``external configuration space" $\mathbb R^{2k+1}_*$, with the non-colliding ones being referred to as {\em MICZ-Kepler orbits}. In this section we shall show that the MICZ-Kepler orbits are all conics. As demonstrated in section \ref{S: Kepler}, that is indeed the case when the magnetic charge is zero, so we shall assume that $\mu\neq 0$ in this section, unless said otherwise. 

Let $(\mathbf r (t), \xi(t))$ be a motion that represents a MICZ-Kepler orbit, and 
\begin{eqnarray}V:= \mathbf r\wedge\mathbf r'\wedge r^3(i\xi, \mathbf r'\lrcorner F) .
\end{eqnarray} Since both $\mathbf r$ and $\mathbf r'$ are orthogonal to $r^3(i\xi, \mathbf r'\lrcorner F)$, we have
\begin{eqnarray}\label{Vnormsq}
|V|^2=|\mathbf r\wedge\mathbf r'|^2\; |r^3(i\xi, \mathbf r'\lrcorner F)|^2={\mu^2\over k}|\mathbf r\wedge\mathbf r'|^4={\mu^2\over k}(|L|^2-\mu^2)^2
\end{eqnarray} by Lemma \ref{lemma} and Eq. (\ref{Lmu}). Therefore $|V|$ is a constant of motion. In view of Eq. (\ref{Lmu}), it vanishes if and only if $\mu= 0$.

\begin{Lem} The $3$-vector $V$ in $\mathbb R^{2k+1}$ is a constant of motion, and it vanishes if and only if $\mu=0$. 
\end{Lem}
\begin{proof} Using the product rule for differentiation, we have
\begin{eqnarray}
V'&=&  \mathbf r\wedge\mathbf r''\wedge r^3(i\xi, \mathbf r'\lrcorner F) +  \mathbf r\wedge\mathbf r'\wedge ((r  \mathbf r')'\lrcorner (i\xi, r^2F))+ \mathbf r\wedge\mathbf r'\wedge (r \mathbf r' \lrcorner (i\xi, r^2F)')\cr
&=& \mathbf r\wedge\mathbf r'\wedge ((r  \mathbf r')'\lrcorner (i\xi, r^2F))+ \mathbf r\wedge\mathbf r'\wedge (r\mathbf r' \lrcorner (i\xi, {D\over dt}(r^2F)))\quad \mbox{by eqn of motion}\cr
&=& \mathbf r\wedge\mathbf r'\wedge ((r  \mathbf r')'\lrcorner (i\xi, r^2F))+ \mathbf r\wedge\mathbf r'\wedge  (r\mathbf r' \lrcorner (i\xi, -\mathbf r\wedge (\mathbf r'\lrcorner F)))\quad \mbox{by Lemma \ref{lemma}}\cr
&=& \mathbf r\wedge\mathbf r'\wedge ((r  \mathbf r')'\lrcorner (i\xi, r^2F))-\mathbf r\wedge\mathbf r'\wedge r^2r' (i\xi, \mathbf r'\lrcorner F) \cr
&=& \mathbf r\wedge\mathbf r'\wedge (r  \mathbf r''\lrcorner (i\xi, r^2F))\cr
&=& \mathbf r\wedge\mathbf r'\wedge (r (i\xi, \mathbf r'\lrcorner F))\lrcorner (i\xi, r^2F))\quad \mbox{by eqn of motion and Lemma \ref{lemma}}\cr
&=& -{\mu^2\over k}\mathbf r\wedge\mathbf r'\wedge \left({\mathbf r\over r}\right)'\quad \mbox{by Lemma \ref{lemma}}\cr
&=& 0.\nonumber
\end{eqnarray}
The rest is clear.
\end{proof}

The nonzero constant decomposable $3$-vector $V$ determines a constant subspace $[V]$ of $\mathbb R^{2k+1}$. The presence of this 3D space $[V]$ is not mysterious because one can show that it is spanned by the three constant vectors $\mathbf{A}$, $\mathbf{A}\lrcorner L$, and $(\mathbf{A}\lrcorner L)\lrcorner L$ when $\mathbf{A}\lrcorner L\neq 0$. Anyhow, it is straightforward to see that $\mathbf A$ is a vector inside the 3D space $[V]$. Since $\mathbf r\wedge V=0$, we know that the MICZ-Kepler orbit is inside the 3D space $[V]$, in fact a conic inside $[V]$, as we shall see in a moment.

Let $\bar L$ be the image of $L$ under the orthogonal projection $\wedge^2 {\mathbb R}^{2k+1}\to \wedge^2 [V]$. Being referred to as the {\bf effective angular momentum}, $\bar L$ shall be seen to play an important role in the study of MICZ-Kepler orbits. 
A simple computation shows that
\begin{eqnarray}
{\bar L}=\left(\mathbf{r}-\frac{r^4}{|L|^2-\mu^2} (i\xi, \mathbf{r}'\lrcorner F)\right)\wedge \left(\mathbf{r}'-\frac{r'}{r} \mathbf{r}\right),
\end{eqnarray}  a decomposable $2$-vector inside the 3D space $[V]$. With the help of Lemma \ref{lemma} and Eq. (\ref{Vnormsq}), one can verify that
\begin{eqnarray}\label{Lbar}
\bar L\wedge \mathbf A ={1\over |L|^2-\mu^2}V, \quad |\bar L|^2 - |\bar L\wedge \mathbf A|^2=|\bar L|^2-{\mu^2\over k}=|L|^2-\mu^2>0.
\end{eqnarray}
From the definition, $\bar L$ should be a constant of motion, a fact which can be verified directly. This fact holds even when $\mu=0$ because $\bar L=L$ in this special case. By some simple computations, one can verify that $\mathbf r\in [V]$ and
\begin{eqnarray} \label{orbit}
r-{\mb A}\cdot {\mb r} = \bar L^2-{\mu^2\over k}, \quad \bar L\wedge ({\mb r} - r\mathbf A)=0.
\end{eqnarray}  

In the following we let $\ms O$ denote the set of all oriented MICZ-Kepler orbits and $\ms M$ denote the set of all pairs $(\mathbf A, \bar L)$ where $\bar L$ is a decomposable 2-vector and $\mathbf A$ is a vector (all inside $\mathbb R^{2k+1}$) such that $|\bar L|^2 > |\bar L\wedge \mathbf A |^2 $. For each
$(\mathbf A, \bar L)\in \ms M$, we claim that $|\bar L- {\mathbf A}\lrcorner({\mathbf A}\wedge \bar L)|>0$, otherwise, $\bar L= {\mathbf A}\lrcorner({\mathbf A}\wedge \bar L)$, after taking the inner product with $\bar L$, we  would have $|\bar L|^2 = |\bar L\wedge \mathbf A |^2 $, a contradiction.

The following theorem extends Theorem 1 of Ref. \cite{meng2012} beyond dimension 3.
\begin{Thm} \label{maintheorem1}

(1) If two motions have the same oriented MICZ-Kepler orbit, then they must have the same effective angular momentum $\bar L$ and the same Lenz vector $\mathbf A$. Consequently one can speak of an oriented MICZ-Kepler orbit with the Lenz vector $\mathbf A$ and the effective angular momentum $\bar L$.

(2) The map $$\varphi_1: \quad \ms O\to \ms M$$ which maps an oriented MICZ-Kepler orbit with the Lenz vector $\mathbf A$ and the effective angular momentum $\bar L$ to $(\mathbf A, \bar L)$ is a bijection. 

(3) An oriented MICZ-Kepler orbit with the Lenz vector $\mathbf A$ and the effective angular momentum $\bar L$ is a conic with its eccentricity $e$ satisfying identity
\begin{eqnarray}
1-e^2={|\bar L|^2- |\bar L\wedge \mathbf A |^2 \over |\bar L-{\mathbf A}\lrcorner({\mathbf A}\wedge \bar L)|^2}(1-|\mb A|^2).
\end{eqnarray}
Moreover, this MICZ-Kepler orbit is oriented in the sense that  $\mathbf t\wedge \mathbf n$ is a positive multiple of $(\bar L-{\mathbf A}\lrcorner({\mathbf A}\wedge \bar L))$. (Here, $\mathbf t$ and $\mathbf n$ are respectively its unit tangent and unit normal vector.) Consequently, reversing the orbit orientation amounts to
turning $(\mathbf A, \bar L)$ into $(\mathbf A, -\bar L)$.

(4) Fix an oriented MICZ-Kepler orbit with the Lenz vector $\mathbf A$ and the effective angular momentum $\bar L$, the total energy for any motion with this oriented MICZ-Kepler orbit is  
\begin{eqnarray}
E=-{1- |\mathbf A|^2\over 2(|\bar L|^2 -|\bar L\wedge \mathbf A|^2)}.
\end{eqnarray}  Consequently one can speak of the total energy $E$ of a MICZ-Kepler orbit. 

(5) A MICZ-Kepler orbit is an ellipse, a parabola and a branch of a hyperbola according as its total energy $E$ is negative, zero and positive.

\end{Thm}
\begin{proof}
\underline{Proof of part (1)}. It is better to use the reformulation for the oriented MICZ-Kepler orbits as oriented conic sections. It shall be shown in the next section that, an oriented MICZ-Kepler orbit for a motion with Lenz vector $\mathbf A$ and effective angular momentum $\bar L$ corresponds to the intersection of the future light cone in the Lorentz space $\mathbb R^{1, 2k+1}$ with the oriented plane 
$$
m\wedge x =0, \quad a\cdot x=1
$$ where 
$$
m={\bar L\wedge A\over |\bar L\wedge A|}, \quad a ={A\over |\bar L\wedge A|^2}
$$ with $A=(1, \mathbf A)$, moreover, $m$ is a decomposable $3$-vector with $m^2=1$, and $a$ is a vector inside the 3D space $[m]$ with $a_0>0$, and $a\lrcorner m$ represents the orientation of the plane. 

Suppose now if a motion with Lenz vector $\mathbf A_1$ and effective angular momentum $\bar L_1$ have the same oriented MICZ-Kepler orbit, then we would have the same oriented plane with this second defining equation:
$$
m_1\wedge x =0, \quad a_1\cdot x=1.
$$ Since this plane does not pass through the origin, it spans both the 3D vector subspace $[m]$ and the 3D vector subspace $[m_1]$, so $m_1=m$ or $-m$. But then we must have $a_1=a$, hence $m_1=m$ because both $a_1\lrcorner m_1$ and $a\lrcorner m$ represent the orientation of the plane. Since $(a, m)$ is uniquely determined by $(\mathbf A, \bar L)$, we have $\mathbf A_1=\mathbf A$ and $\bar L_1=\bar L$. In summary, the oriented plane, hence the oriented MICZ-Kepler orbit, is uniquely determined by the pair $(\mathbf A, \bar L)$. 

\underline{Proof of part (2)}. From part $(1)$ we have a well-defined map
 $$\varphi_{1}:  \mathscr{O}\rightarrow \mathscr{M},$$
 which sends each oriented MICZ-Kepler orbit to the pair $( \mathbf{A}, \bar{L})$ consisting of its unique Lenz vector $\mathbf A$ and effective angular momentum $\bar L$. It is clear from the proof of part (1) that $\varphi_{1}$ is one-to-one. It remains to show that $\varphi_{1}$  is onto, i.e., for a given pair $( \mathbf{A}, \bar{L})\in \ms M$, we need to find an initial data consisting of an initial position $\mathbf q$, an initial velocity $\mathbf v$ (which shall be chosen such that $\mathbf v \cdot \mathbf q=0$) and an initial point $\eta$ in an magnetic orbit such that
\begin{eqnarray}\label{check}
\left\{
\begin{array}{rcl}
\bar L &= & \mathbf q\wedge \mathbf v+{1\over |\mathbf v|^2}\mathbf v\wedge  (\mathbf v\lrcorner (i\eta, |\mathbf q|^2 F(\mathbf q)))\\
\\
\mathbf A & = & \mathbf v\lrcorner \bar L +{\mathbf q\over |\mathbf q|}. 
\end{array}\right.
\end{eqnarray}

Just as in the proof of part (3) of Theorem 1 in Ref. \cite{meng2012}, a key step is to find a {\em unit} vector $\mathbf n$ such that
$$
\bar L\wedge \mathbf n = \bar L\wedge \mathbf A, \quad\mbox{but}\quad \mathbf A\neq \mathbf n\quad \mbox{and $|\mathbf A-\mathbf n|$ is maximal possible}.
$$ The existence of such a unit vector $\mathbf n$ is guaranteed by the condition $|\bar L|^2 > |\bar L\wedge \mathbf A|^2$. To see this, we write $\mathbf A$ as the sum of the vector $\mathbf A_{||}\in [\bar L]$ and  vector $\mathbf A_{\perp}$ perpendicular to $[\bar L]$. Then the condition $|\bar L|^2 > |\bar L\wedge \mathbf A|^2$ is just $|\bar L|^2 > |\bar L|^2 |\mathbf A_\perp|^2$, which implies that $|\mathbf A_\perp|<1$. Now we just take $\mathbf n = \mathbf n_{||}+\mathbf A_\perp$ where $\mathbf n_{||}\in [\bar L]$ such that 
$|\mathbf n_{||}|^2 + |\mathbf A_\perp|^2 =1$, $\mathbf n_{||}\wedge \mathbf A_{||} =0$ and $\mathbf n_{||}\cdot \mathbf A_{||}\le 0$.   Clearly $\mathbf n$ is unique unless
$\mathbf A \lrcorner \bar L =0$. 

We are now ready to find the initial data. First, we let
\begin{eqnarray}
\mathbf v: ={1\over |\bar L|^2}(\mathbf n-\mathbf A) \lrcorner\bar L.
\end{eqnarray}

Since $\mathbf A-\mathbf n\neq 0$ and $(\mathbf A-\mathbf n)\wedge \bar L=0$, we have $\mathbf v\neq 0$, so we can let
\begin{eqnarray}
\mathbf q :={1-\mathbf A\cdot \mathbf n\over |\mathbf v|^2}\mathbf n.
\end{eqnarray}
Since $1-\mathbf A\cdot \mathbf n=1-|\mathbf A_\perp|^2-\mathbf n_{||}\cdot \mathbf A_{||}>0$, we have $\mathbf q\neq 0$ and $\mathbf q\cdot \mathbf v=0$, as we promised a few paragraphs above.  

With this choice of $\mathbf v$ and $\mathbf q$, one can check that the second identity in Eq. (\ref{check}) holds.
One can also check that, when $\bar L\wedge \mathbf A=0$, the first identity in Eq. (\ref{check}) holds if we
take $\eta:=0$. It remains to find $\eta$ for the case $\bar L\wedge \mathbf A\neq 0$ so that the first identity in Eq. (\ref{check}) holds.

When $\bar L\wedge \mathbf A\neq 0$, the vertical component $\mathbf n_\perp$ of $\mathbf n$ is nonzero, so
\begin{eqnarray}
\mathbf u:=\mathbf v\lrcorner \bar L+(1-\mathbf A\cdot \mathbf n)\mathbf n\neq 0.
\end{eqnarray} One can check that the nonzero vectors $\mathbf u$, $\mathbf v$ and $\mathbf n$ are mutually orthogonal, so we can assume that $\mathbf n=\mathbf e_{2k+1}$, $\mathbf v=|\mathbf v|\mathbf e_{2k}$, and $\mathbf u=|\mathbf u|\mathbf e_{2k-1}$ or $-|\mathbf u|\mathbf e_{2k-1}$.  By letting
\begin{eqnarray}
\eta := {|\mathbf u|\over |\mathbf v|}(M_{12}+\cdots +M_{2k-3, 2k-2}+\mathrm{sign}\,(\mathbf u\cdot \mathbf e_{2k-1}) M_{2k-1, 2k}),
\end{eqnarray} one can check that $\mathbf v\lrcorner (i\eta, |\mathbf q|^2 F(\mathbf q))=\mathbf u$ with the help of Eq. (\ref{Ffield}), therefore, the right-hand side of the first identity in Eq. (\ref{check}) becomes
$$
\mathbf q\wedge \mathbf v+{1\over |\mathbf v|^2}\mathbf v\wedge \mathbf u = {1\over |\mathbf v|^2}\mathbf v\wedge (\mathbf v\lrcorner \bar L) = \bar L,
$$ i.e., the left-hand side of the first identity in Eq. (\ref{check}).

\underline{Proof of parts (3), (4) and (5)}. Part (3) follows by adapting the proof for part (1) of Theorem 1 in Ref. \cite{meng2012}. Part (4) follows from Eqs. (\ref{formulaE}) and (\ref{Lbar}). Part (5) is a consequence of parts (3) and (4). 

\end{proof}
\section{Oriented MICZ-Kepler orbits and the Lorentz group}\label{S: Main2}
The goal here is to relate the MICZ-Kepler orbits to the Lorentz group, a phenomena initially found by the second author \cite{meng2012} for dimension three. 

As in Ref. \cite{meng2012}, the key is to go to the light cone formulation for the oriented MICZ-Kepler orbits. To do that, we write $x=(x_0, \mathbf r)$, $A=(1, \mathbf A)$, and verify by computation that $$(\bar L\wedge A, \bar L\wedge A)^2 = |L|^2-\mu^2,$$
a positive number, so $\bar L\wedge A\neq 0$. Then the MICZ-Kepler orbit as defined in Eq. (\ref{orbit}) can be reformulated as the intersection of the future light cone inside $\mathbb R^{1, 2k+1}$ with the affine plane
\begin{eqnarray}\label{conics3}
x\wedge (\bar L\wedge A)=0, \quad A\cdot x =|\bar L\wedge A|^2.
\end{eqnarray} This description of the MICZ-Kepler orbit is also valid when $\mu =0$, because in this special case $\bar L =L$ and Eq. (\ref{conics3}) is just Eq. (\ref{conics1}). 

We note that the affine plane defined by Eq. (\ref{conics3}) passes through the point $\bar L\lrcorner(\bar L\wedge A)$ and is parallel to the 2D subspace $[A\lrcorner (\bar L\wedge A)]$ of $\mathbb R^{1, 2k+1}$, that is because Eq. (\ref{conics3}) is equivalent to equation
\begin{eqnarray}\label{conics4}
(x-\bar L\lrcorner(\bar L\wedge A))\wedge (A\lrcorner (\bar L\wedge A))=0.
\end{eqnarray}

Since this affine plane is oriented by $A\lrcorner (\bar L\wedge A)$ and has a positive $x_0$-intercept, its intersection with the future light cone is an oriented conic section. The natural projection from $\mathbb R^{1, 2k+1}$ onto $\mathbb R^{2k+1}$ provides an orientation preserving diffeomorphism from this oriented conic section onto the oriented MICZ-Kepler orbit as defined via Eq. (\ref{orbit}).

\vskip 10pt
It is geometrically more convenient to rewrite equation (\ref{conics3}) as
\begin{eqnarray}\label{conics2}
m\wedge x =0,\quad a\cdot x = 1
\end{eqnarray} and equation (\ref{conics4}) as
\begin{eqnarray} 
(x-{(e_0\lrcorner m)\lrcorner m\over a_0})\wedge (a\lrcorner m)=0.
\end{eqnarray}
where $m={\bar L\wedge A\over |\bar L\wedge A|}$, $a ={A\over |\bar L\wedge A|^2}$. It is easy to see that the pair $(a, m)$ thus obtained satisfies the following condition:
{\em $m$ is a decomposable $3$-vector in the Lorentz space $\mathbb R^{1, 2k+1}$ such that $m^2=1$, and $a=(a_0, \mathbf a)$ is a vector in the Lorentz space $\mathbb R^{1, 2k+1}$  with $a_0>0$ and $a\wedge m=0$}. 

To proceed, we let $\mathcal M$ be the set of all pairs $(a, m)$ with the conditions we have just specified in the last paragraph. We remark that an element of $\mathcal M$ is just an oriented 3D Lorentz subspace [m] of $\mathbb R^{1, 2k+1}$ together with a vector $a\in [m]$ with a positive temporal component, moreover, the map $(a, m)\mapsto m$ is a fiber bundle map
with the $\mathcal M$ as the total space, the upper half 3D Lorentz space
$$\mathbb R^{1,2}_+:=\{(x_0, \mathbf x)\in \mathbb R^{1,2} \mid x_0>0\}$$
as the fiber, and the space $\widetilde{\mathrm{Gr}}_{1,2}(\mathbb R^{1, 2k+1})$ consisting of oriented 3D Lorentz subspaces in $\mathbb R^{1, 2k+1}$ as the base space.

For any $(a, m)\in \mathcal M$, we note that $a\lrcorner m\neq 0$, a fact one may check under the assumption that $m=e_0\wedge e_1\wedge e_2$. It is not hard to see that the map
\begin{eqnarray}\varphi_2: \quad\mathcal{M}&\longrightarrow &\mathscr{M}\\
(a, m)&\mapsto& \left (\frac{\mathbf{a}}{a_{0}}, \frac{e_{0}\lrcorner m}{\sqrt{a_{0}}}\right)
\end{eqnarray}
is a bijection whose inverse maps $(\mathbf{A}, \bar L)$ to
$$
\left({A\over |\bar L\wedge A|^2}, {\bar L\wedge A\over |\bar L\wedge A|}\right)
$$ where $A=(1, \mathbf A)$. By composing this bijection $\varphi_2$ with the bijection $\varphi_1$ in the part (2) of Theorem \ref{maintheorem1}, we get a bijection 
\begin{eqnarray}\label{bijection}
\varphi:\quad \ms O\to \mathcal M.
\end{eqnarray}

The following theorem extends Theorem 2 of Ref. \cite{meng2012} beyond dimension 3.
\begin{Thm}\label{maintheorem2}
\label{main2} Let $\ms O_+$ ($\ms O_0$ reps.) be the the set of oriented elliptic (parabolic reps.) MICZ-Kepler orbits, ${\mathcal M}_+ =\{(a, m)\in {\mathcal M}\mid a^2 > 0\}$,  and ${\mathcal M}_0 =\{(a, m)\in {\mathcal M}\mid a^2 = 0\}$.

(1) For the oriented MICZ-Kepler orbit  parametrized by $(a, m)\in \mathcal M$, we have the energy formula
$$
E=-{a^2\over 2a_0}.
$$

(2) The bijection $\varphi$ in Eq. (\ref{bijection}) maps ${\ms O}_+$ onto $\mathcal M_+$  and ${\ms O}_0$ onto $\mathcal M_0$.

(3) The action of       $~{\mr {SO}}^+(1,2k+1)\times {\bb R}_+$ on $\mathcal M_0$ defined by $$(\Lambda, \lambda)\cdot (a, m) = (\lambda\cdot (\Lambda a), \Lambda m)$$ is transitive. So ${\mr {SO}}^+(1,2k+1)\times {\bb R}_+$, in fact ${\mr {SO}}^+(1,2k+1)$ also,  acts transitively on the set of oriented parabolic MICZ-Kepler orbits.

(4) The action of   $~{\mr {SO}}^+(1,2k+1)\times {\bb R}_+$ on $\mathcal M_+$ defined by $$(\Lambda, \lambda) \cdot (a, m) = (\lambda\cdot \Lambda a, \Lambda m)$$ is transitive. So ${\mr {SO}}^+(1,2k+1)\times {\bb R}_+$ acts transitively on the set of oriented elliptic MICZ-Kepler orbits.

(5) The action in parts (3) and (4) extends to ${\mr {O}}^+(1,2k+1)\times {\bb R}_+$.

\end{Thm}
\begin{proof} Only the proof of transitivity in parts (3) and (4) needs a little effort. It is easy to see that, there are mutually orthogonal vectors $u$, $v$ and $w$ in $\mathbb R^{2k+1}$ such that $|u|<1$, $v\wedge w\neq 0$, and $m=(e_0+u)\wedge v\wedge w$. After a rotation of $\mathbb R^{2k+1}$, we may assume that $m$ is a decomposable $3$-vector in the Lorentz subspace $\mathbb R^{1,3}$. Then, the proof of transitivity descends to the proof of transitivity in Theorem 2 of Ref. \cite{meng2012}.

\end{proof}
We conclude this article with a remark: while the Lenz vector and the angular momentum determine an oriented orbit, it is the Lenz vector and the effective angular momentum that determine an oriented orbit in the ``external configuration space"; moreover, each oriented orbit is the lifting via the canonical connection (i.e., the generalized Dirac monopole) of an oriented orbit in the ``external configuration space".

\end{document}